\documentclass[runningheads]{llncs}
\usepackage{graphicx}
\usepackage{cite}
\usepackage{mathrsfs}
\usepackage{amsfonts}
\usepackage{amsmath}
\usepackage{amsfonts,amssymb,color,enumerate}
\usepackage{dsfont}
\usepackage{curves}
\usepackage{cases}
\usepackage{mathrsfs}
\usepackage{pifont}
\usepackage{amssymb}
\usepackage{epsfig}
\usepackage{bm}
\usepackage{epstopdf}
\usepackage{graphicx,color}
\usepackage{algorithm}
\usepackage{algorithmic}
\usepackage{color}
\usepackage{float}
\usepackage{bbding}

\begin{document}
\title{Approximation algorithms for general cluster routing problem}
%
%
\author{Xiaoyan Zhang\inst{1} \and
Donglei Du\inst{2}\and
Gregory Gutin\inst{3}\and
Qiaoxia Ming\inst{1}\and
Jian Sun\inst{1}(\Envelope)\thanks{This research is supported or partially supported by the
National Natural Science Foundation of China (Grant Nos. 11871280,
11371001, 11771386 and 11728104), the Natural Sciences and Engineering Re-search Council of Canada (NSERC) Grant 06446 and Qinglan Project.}}

\authorrunning{Xiaoyan Zhang et al.}
%
\institute{School of Mathematical Science \& Institute of Mathematics, Nanjing Normal University, Jiangsu 210023, P. R. China \\
\email{sunjian199203@126.com}\and
Faculty of Management, University of New Brunswick, Fredericton, New
Brunswick, Canada, E3B 5A3\\
\email{ddu@unb.ca}\and
Department of Computer Science
Royal Holloway, University of London
Egham, Surrey, TW20 0EX, UK\\
\email{g.gutin@rhul.ac.uk}
}
\maketitle              
\begin{abstract}
Graph routing problems have been investigated extensively in operations research, computer science and engineering due to their ubiquity and vast applications. 
In this paper, we study constant approximation algorithms for some variations of the general cluster  routing  problem. In this problem, we are given an edge-weighted complete undirected graph $G=(V,E,c),$ whose vertex set is partitioned into clusters $C_{1},\dots ,C_{k}.$ We are also given a subset $V'$ of $V$ and a subset $E'$ of $E.$ The weight function $c$ satisfies the triangle inequality.
The goal is to find a minimum cost walk $T$ that visits each vertex in $V'$ only once, traverses every edge in $E'$ at least once and for every $i\in [k]$ all vertices of $C_i$ are traversed consecutively.
\end{abstract}
\section{Introduction}\label{Introduction}
Graph routing problems have been studied extensively since the early 1970s. Most of there problems are NP-hard, and hence no polynomial-time exact algorithms exist for most of them unless P=NP. In a typical routing problem, a salesman starts from a home location, visits a set of prescribed cities exactly once, and returns to the original location with minimum total distance travelled.


Arguably the most well-known routing problem is the \emph{travelling salesman problem (TSP)} (see \cite{Gutin} for a compendium of results on the problem). We are given a weighted graph $G=(V,E, c)$ (directed or undirected) with vertex set $V$, edge set $E$, and cost $c(e)$ for each edge $e\in E$.  The TSP's goal is to find a Hamiltonian cycle with minimum total cost.
Without loss of generality, we may assume that $G$ is a complete graph (digraph); otherwise, we could replace the missing edges with edges of very large cost.





Unfortunately, the TSP is NP-hard even for metric arc costs \cite{Karp}. Therefore, one approach for solving the TSP (and other NP-hard problems) is using (polynomial-time) approximative algorithm whose performance is measured by the {\em approximation ratio}, which is the maximum ratio of the approximative solution value to the optimum value among all problem instances. The best known approximation algorithm for the TSP with triangle inequality is by Christofides \cite{Christofides} with ratio 1.5. For the general TSP where the triangle inequality does not hold, there is no (polynomial-time) approximation algorithm with a constant approximation ratio,
unless P=NP \cite{Sahni}. TSP along with its variations have been extensively investigated in the literature. Here are two generalizations of TSP studied in the literature.

The \emph{general routing problem (GRP)}: Let $G=(V,E,c)$ be an edge-weighted complete undirected graph such that the triangle inequality holds for the weight function $c$. The goal is to find a minimum cost
walk that visits  each vertex in a required subset $V^{'}\subseteq V$ exactly once and traverses every edge in a required subset $E^{'}\subseteq E$ at least once. For this problem, Jansen \cite{Jansen} gave a 1.5-approximation algorithm.

The \emph{ cluster travelling salesman problem (CTSP)}: Let $G=(V,E,c)$ be an edge-weighted complete undirected graph such that the triangle inequality holds for the weight function $c$. The vertex set $V$ is partitioned into clusters $C_{1},\dots ,C_{k}$. The goal is to compute a minimum cost Hamiltonian cycle $T$ that visits all vertices of each cluster consecutively (and thus for each cluster we have starting and finishing vertices on $T$). Arkin et al. \cite{Arkin} designed a 3.5-approximation algorithm for the problem with given starting vertices in each cluster. Guttmann-Beck et al. \cite{Guttmann} proposed a 1.9091-approximation algorithm for the problem in which the starting and ending vertices of each cluster are specified and gave a 1.8-approximation algorithm if for each cluster two vertices are given such that one of the them can be a starting vertex and the other the finishing vertex.

In this paper, we introduce and study the \emph{general cluster routing problem} (GCRP) which generalizes both GRP and CTSP.
We provide approximation algorithms of constant approximation ratio for variations of this problem.
%
%
 In GCRP, we are given an edge-weighted undirected graph $G=(V,E, c)$  such that the triangle inequality holds for the weight function $c$.  The vertex set $V$ is partitioned into clusters $C_{1},\dots ,C_{k}$. For any given  vertex subset $V^{\prime}\subseteq V$ and edge subset $E^{\prime}\subseteq E$, the aim is to find a minimum cost walk $T$ (hereafter a walk will be called a {\em tour})  that visits each vertex in $V^{\prime}$ exactly once and traverses each edge in $E^{\prime}$ at least once such that for every $i\in [k]$ all vertices of $T$ belonging to $C_i$ are visited consecutively in $T$.
 Depending on whether or not the starting and finishing vertices of a cluster are specified or not, we consider two cases. When every cluster has a pair of specified starting and finishing vertices, we offer a 2.4-approximation combinatorial algorithm. When every cluster has unspecified starting and finishing vertices, depending on whether the {\em required} edges (i.e., those in $E'$) are incident with different clusters or not, we further consider two subcases. If  all required edges are only distributed in the clusters, we get a 3.25-approximation combinatorial algorithm. On the other hand, if there exist edges from $E^{\prime}$ incident with two different clusters, we get a 2.25-approximation combinatorial algorithm.

The remainder of this paper is organized as follows. 
We provide some preliminaries in Section~\ref{sec:prel}. We study algorithms for the 
GCRP in Sections~
\ref{sec:cgrp}. 
We conclude in Section \ref{sec:c}.
Every theorem whose proof is given in Appendix is marked by ($\star)$.

\section{Preliminaries}\label{sec:prel}
In this section, we recall some algorithms for three problems along with some preliminary results, which will be used as subroutines in our algorithms later.
\subsection{The Travelling Salesman Path Problem}
The \emph{traveling salesman path problem (TSPP)} \cite{Fumei, Traub, Hoogeveen, Guttmann, Sebo1, Sebo2} is a generalization of the TSP, but received much less attention  than TSP in the literature.
In TSPP, given an  edge-weighted undirected graph $G = (V,E,c)$ and two vertices $s, t\in V$, the aim is to find a minimum cost Hamiltonian path from $s$ to $t.$  Note that vertices $s$ and $t$ need not be distinct. However, when $s=t$ TSPP is equivalent to the TSP. 
Let $MST(G)$ be a minimum spanning tree of $G$. For simplicity, $MST(G)$ will also denote the cost of this tree.

Hoogeveen \cite{Hoogeveen} considered three variations of the travelling salesman path problem (TSPP), where as part of the inputs, the following constraints are placed on the end vertices of the resulting Hamiltonian path:
\begin{enumerate}[(1)]
\item both the source and the destination are specified;
\item one of the the source and the destination is specified;
\item neither the source nor the destination are specified.
\end{enumerate}
\begin{property}
For Cases (2) and (3), it was shown in \cite{Hoogeveen} that a straightforward adaptation of Christofide's algorithm can yield an algorithm with a performance ratio of $\frac{3}{2}$.
\end{property}

However, Case (1) is more difficult, for which many results exist in the literature. On the positive side, a $\frac{5}{3}$-approximation algorithm is proposed in \cite{Hoogeveen}, followed by an improved $\frac{8}{5}$-approximation in \cite{Sebo2}. Sebo \cite{Sebo1} gave a strongly polynomial algorithm and improved the analysis of the metric $s-t$ path TSP. He found a tour of cost less than 1.53 times the optimum of the subtour elimination LP.  On the negative side, the usual integer linear programming formulation has an integrality gap at least 1.5.

Let $c(P)$ be the sum of all the edge costs of a given path or tour $P$. The following result from \cite{Hoogeveen} will be used later.
\begin{theorem}\label{th12} \cite{Hoogeveen}
There exists a polynomial-time algorithm for travelling salesman path problem with given end
vertices $s$ and $t$, and we can find two solutions $S_{1}$ and $S_{2}$ for the problem which satisfy the following inequalities:
\begin{eqnarray*}
c(S_{1}) &\le& 2 MST(G)-c(s,t)\leq 2 OPT-c(s,t),\\
c(S_{2}) &\le& MST(G)+\frac12(OPT+c(s,t))\le\frac 32 OPT+\frac12c(s,t).
\end{eqnarray*}
\end{theorem}
\begin{corollary}\cite{Hoogeveen}\label{co1}
The shorter of the tours $S_{1}$ and $S_{2}$ is at most $\frac{5}{3}OPT$.
\end{corollary}
\begin{proof} By Theorem \ref{th12}, if $c(s,t)\geq \frac{1}{3}OPT$, then $c(S_{1})\leq\frac{5}{3}OPT$. Otherwise (i.e. $c(s,t)\leq \frac{1}{3}OPT$) we have $c(S_{2})\leq\frac{5}{3}OPT$.\qed
\end{proof}

Below, we consider a more general problem, called the \emph{travelling general path problem (TGPP)}. Let $G=(V,E,c)$ be a weighted connected graph with two specified ending vertices $s, t\in V$. For any given vertex subset $V^{\prime}\subseteq V$ and edge subset $E^{\prime}\subseteq E$, the objective is to find a minimum cost path from $s$ to $t$ in $G$ that visits all vertices in $V^{\prime}$ exactly once and traverses all edges in $E^{\prime}.$ Note that when $s=t$, this problem becomes the general routing problem introduced in \cite{Bienstock} which was discussed earlier. We focus on the case $s\neq t$ in the reminder of this paper.

Note that this is a minimum cost problem and the edge costs satisfy the triangle inequality. Thus, we can reduce the visits of vertices and edges not in $V^{\prime}$ and $E^{\prime}$. Namely, we can create a new reduced graph as follow in the problem:
\[
G^{\prime}=(\{v|v\in e, e\in E^{\prime}\}\cup\{s\}\cup\{t\}\cup V^{\prime},E^{\prime}).
\]

We assume that $s$ and $t$ are two different vertices in the new graph $G^{\prime}$. First, we compute the connected components of $G^{\prime}$ via depth-first search in polynomial time. Then, contracting each component to a vertex, we construct a new complete graph $G^{*}$, where each edge cost between vertices is the longest edge cost between each pair of components, which is defined as the distance of each pair of component. This can be done in polynomial time. But we only consider those edges between the vertices with degree $d(v)\in\{0,1\}$.  Finally from the graph $G^{*}$, we create a feasible solution as described in Algorithm \ref{alg1}.
\begin{algorithm}[!htb]
	\renewcommand{\algorithmicrequire}{\textbf{Input:}}
    \renewcommand{\algorithmicensure}{\textbf{Output:}}
	\caption{Algorithm of TGPP with specified vertice}
	\begin{algorithmic}[1]\label{alg1}
		\REQUIRE
 \STATE An edge-weighted undirected graph $G=(V,E,c)$.
 \STATE Starting vertex $s$ and ending vertex $t$ of $G$.
 \STATE $V^{\prime}\subseteq V$, $E^{\prime}\subseteq E$ are required vertex subset and edge subset, respectively.

        \ENSURE A travelling general salesman path.


\textbf{begin}:\\
 \STATE Construct  a new graph $G^{\prime}=(\{v|v\in e, e\in E^{\prime}\}\cup\{s\}\cup\{t\}\cup V^{\prime},E^{\prime})$.
 \STATE Compute the connected components $K_{1},\dots,K_{k}$ of $G^{\prime}$.
 \STATE Let $U$ be the set of vertices $v$ with degree $d(v)\in\{0,1\}$. Define a complete graph $G_{k}=([k], E_{k})$ with the cost $c(e)$ of edge $e=(i,j)$ with $i\neq j$ equal to the longest link between a vertex in $K_{i}\cap U$ and a vertex in $K_{j}\cap U.$
 \STATE Copy the edges of $MST(G^{*})$ except for those on $s$-$t$ path.
 \STATE Find an Eulerian walk between $s$ and $t$.
 \STATE Turn the Eulerian walk into a Hamilton path $S.$
 \STATE \textbf{output} $S$.\\
\textbf{end}\\
	\end{algorithmic}
\end{algorithm}

\begin{theorem}\label{th4}($\star$)
Let $S$ be the path output by Algorithm \ref{alg1}. Then we have
\begin{align*}
c(S)\le \min\left\{3OPT-c(s,t), \frac{3}{2}OPT+\frac{1}{2}c(s,t)\right\}.
\end{align*}
\end{theorem}
\begin{corollary}
The length of the tour output by Algorithm \ref{alg1} is at most $2OPT$.
\end{corollary}
\begin{proof}
By Theorem \ref{th4}, if $c(s,t)\geq OPT$, then $c(S)\leq 3OPT-c(s,t)\leq 2OPT$. Otherwise, if $c(s,t)\leq OPT$, we have $c(S)\leq\frac{3}{2}OPT+\frac12c(s,t)\leq 2OPT$.  \qed
\end{proof}

\subsection{The Stacker Crane Problem}
Given a weighted graph $G=(V,E,c)$ whose edge costs satisfy the triangle inequality. Let $D=\{(s_{i},t_{i}):i=1,\dots ,k\}$ be a given set of special directed arcs, each with length $l_{i}$. The arc $\overrightarrow{(s_{i},t_{i})}$ denotes an object that is at vertex $s_{i}$ and needs to be moved to vertex $t_{i}$ using a vehicle (called the stacker crane). The problem is to compute a shortest walk that traverses each directed arc $\overrightarrow{(s_{i},t_{i})}$ at least once in the specified direction (from $s_{i}$ to $t_{i}$). Let $D=\sum\limits_{i}l_{i}$ and $A=OPT-D$.

This problem is a generalization of the TSP, which can be viewed as an instance of this problem where each vertex is replaced by an arc of zero-length. Frederickson et al. presented a 1.8-approximation algorithm for this problem \cite{Frederickson}. This algorithm applies two subroutines and then selects the better of the two solutions generated. The main ideas of these two subroutines are summarized below for convenience (see \cite{Frederickson, Guttmann} for details):
\begin{itemize}
  \item Algorithm Short-Arcs 1: Shrink the directed arcs and reduce the problem to an instance of TSP. Use an approximation algorithm for the TSP instance, and then recover a solution for the original problem. This algorithm works well when $D\leq\frac{3}{5}OPT$.
  \item Algorithm Long-Arcs 1: Complete the set of directed arcs into a directed cycle cover. Then find a set of edges of minimum total weight to connect the cycles together. Add two copies of each one of these edges, and orient the copies in opposite directions to each other. The resulting graph is Eulerian, and the algorithm outputs an Euler walk of this solution. The algorithm performs well when $D>\frac{3}{5}OPT$.
\end{itemize}
The following theorem can be derived from \cite{Frederickson}.
\begin{theorem}\label{th5}\cite{Frederickson}
Consider an instance of the Stacker Crane Problem where the sum of the lengths of the special directed arcs is $D$. Let $OPT$ be an optimal solution, and let $A=OPT-D$. The walk returned by Algorithm Short-Arcs 1 has length at most $\frac{3}{2}A+2D$. The walk returned by Algorithm Long-Arcs 1 has length at most $3A+D$.
\end{theorem}
\subsection{The Rural Postman Problem}
 Let $E^{\prime}\subseteq E$ be a specified subset of special edges. We use $c(e)$ to denote the edge cost of $e$. The rural postman problem (RPP) is to compute a shortest walk that visits all the edges in $E^{\prime}$. The Chinese Postman Problem is a special case of RPP in which $E^{\prime}=E$, i.e., the walk must include all the edges. The Chinese Postman Problem is solvable in polynomial time by reducing it to weighted matching, whereas RPP is NP-hard. Let $D=\sum_{i}l_{i}$ be the total length of the paths in all clusters. We recall the algorithms in \cite{Frederickson, Guttmann}.
 \begin{itemize}
   \item Algorithm Short-Arcs 2: Consider the line graph $c(G)$ of original graph $G$. This algorithm works well when $D\leq\frac{3}{5}OPT.$
   \item Algorithm Long-Arcs 2: Complete the set of undirected arcs into a cycle cover. Then find a set of edges of minimum total weight to connect the cycles together. Add two copies of each one of these edges. The resulting graph is Eulerian, and the algorithm outputs an Euler walk of this solution. The algorithm performs well when $D$ is large. Note that Algorithm Long-Arcs 2 is similar to Long-Arcs, but in this case, $D$ is a set of undirected edges. The algorithm performs well when $D>\frac{3}{5}OPT$.
 \end{itemize}
 The two algorithms defined above for SCP can be modified to solve RPP. It is easy to see that the second part of Theorem \ref{th5} holds for this case as well, i.e. the walk returned by Algorithm Long-Arcs 2 has length at most $3A+D$.
\begin{remark}\label{re1}
As indicated by Frederickson et al. \cite{Frederickson}, it is easy to show that the above algorithms produce a $\frac{3}{2}$ performance ratio for RPP.
\end{remark}

\section{The general cluster routing problem}\label{sec:cgrp}

\subsection{The general cluster routing problem with pre-specified starting and ending vertices}\label{subsec:cgrp_nse}

Note that there may exist two subcases in this case. First, each edge in $E^{\prime}$ is fully contained in its cluster. Second, some edges may be incident with more than one cluster. 

Let $s_i$ and $t_i$ be pre-specified starting and ending vertice of cluster $C_i$, $i\in [k].$ Since the goal is to find a minimal total edge cost and the edge costs satisfy the triangle inequality, we can ignore the vertices not in $V^{\prime}$ and edges not in $E^{\prime}$ from graph $G$ to consider a new graph instead. Namely, for every cluster $C_{i}$, $i\in [k]$, consider the GCRP in the following new graph $\overline{G}=\cup \overline{C_{i}}$, where
\begin{equation*}
\overline{C_{i}}=(\overline{V_{i}},\overline{E_{i}})=\Bigg{(}\{v|v\in e, e\in E^{\prime}_{i}\}\cup V^{\prime}_{i}\cup \{s_{i}\}\cup \{t_{i}\}, E^{\prime}_{i}\Bigg{)}.
\end{equation*}
Our algorithm is based on the following idea. First, within each cluster $\overline{C_{i}}$, we find a path $p_{i}$, starting with $s_{i}$ and ending at $t_{i}$, visits all the vertices in $V^{\prime}$ and edges of each cluster $\overline{C_{i}}$. This can be done by Algorithm \ref{alg1}. Second, we need to connect the paths by adding some edges to make the resulting graph into a single cycle.

Let $G=(V,E)$ be a complete graph with vertex set $V$ and edge set $E$, the vertex set is partitioned into clusters $C_{1},\dots ,C_{k}$. The starting and ending vertices in each cluster are specified. Let $\overline{C_{i}}=(\overline{V_{i}},\overline{E_{i}})$ be the new graph as described above. Clearly, the desired tour in $G$ does not always exist, e.g., when there exists a required edge $e\in E^{\prime}$ between cluster $C_{i}$ and cluster $C_{j}$, $i\neq j$, and this required edge is not a $(t_{i},s_{j})$ edge (in such a case,  at least one of the clusters must be visited more than one time). Henceforth, we will assume that the desired tour does exist.

\begin{algorithm}[H]
	\renewcommand{\algorithmicrequire}{\textbf{Input:}}
    \renewcommand{\algorithmicensure}{\textbf{Output:}}
	\caption{Algorithm of given starting and ending vertices}
	\begin{algorithmic}[1]\label{alg3}
		\REQUIRE
\STATE An edge-weighted graph $G=(V,E,c)$.
\STATE A partition of $V$ into clusters $C_{1},\dots,C_{k}$.
\STATE Each cluster $C_{i}$ with starting and ending vertices $s_{i}$ and $t_{i}$, respectively, $i=1,\dots,k$.

        \ENSURE A cluster general routing tour.

\textbf{begin}:\\
 \STATE Construct a new graph $\overline{G}=\cup_{i=1}^{k} \overline{C_{i}}$.
\STATE For $i=1,\dots, k$, apply Algorithm \ref{alg1} to get a path $p_{i}$ and orient the $(s_{i},t_{i})$ edge a direction, from $s_{i}$ to $t_{i}$, to obtain the arc $\overrightarrow{(s_{i},t_{i})}$.
\STATE Apply Algorithm Short-Arcs 1 and Algorithm Long-Arcs 1 for SCP on special arc $\overrightarrow{(s_{i},t_{i})}$, $i=1,\dots,k$, and output the shorter solution $T$.
\STATE In $T$, replace the special directed arc $(s_{i},t_{i})$ by the path $p_{i}$, for $i=1,\dots , k$.
\STATE \textbf{Output} the resulting tour $T_{s}$.\\
\textbf{end}\\
	\end{algorithmic}
\end{algorithm}

The main idea of Algorithm \ref{alg3} is illustrated as follows:

In Step 1, we first consider the number of connected components of $\overline{C_{i}}$. If the number is 1, it means that there exists a path from $s_{i}$ to $t_{i}$ that visits all the required edges in $E^{\prime}$ and vertices in $V^{\prime}$. When the number is more than 2, shrinking the connected components to be vertices and finding a path to connect all these vertices lead to a feasible solution.

In Step 2, by applying Algorithm \ref{alg1}, we can get a path from the starting vertex $s_{i}$ to ending vertex $t_{i}$.

In Step 3, we only need to connect these clusters to form a tour. In this progress, we can shrink the directed arc $\overrightarrow{(s_{i},t_{i})}$ and reduce the problem to an instance of TSP. Use Christofides' algorithm \cite{Christofides} for the TSP instance.

In Step 4 by replacing the special directed arc $(s_{i},t_{i})$ by path $p_{i}$, we obtain a solution to the original graph.

Our algorithm is a combinational algorithm, which deals with the condition of  the pre-specified starting and ending vertices carefully. Let $OPT$ be the cost of the optimal solution. Let $L$ be the sum of lengths of the paths of $OPT$ through each cluster and let $A$ be the length of the other edges of $OPT$ that are not in $L$. Let $D$ be the total length of the directed arcs $(s_{i},t_{i}),i=1,\dots,k$. Then, we have the following theorem:
\begin{theorem}\label{th6}($\star$)
 Let $T$ be the tour output by Algorithm \ref{alg3}. Then
$$c(T_{s})\leq2.4\:OPT.$$
\end{theorem}
For the second case, there exists required edges in $E^{\prime}$ between different clusters. If there exist required edges incident with two different clusters, they must be $(t_{i},s_{j})$ edges. First, we need to compute the number of $(t_{i},s_{j})$ edges. Suppose the number is $k$. If $k=0$, it is just Case 1. If $k\geq2$, we then get $k+1$ components and we can shrink the components and go back to Case 1 again.

According to Theorem \ref{th6}, for the general cluster routing problem with pre-specified vertices, we now get a 2.4-approximation combinatorial algorithm.

\subsection{The general cluster routing problem without specifying starting and ending vertices}
In this section, we consider the version of GCRP where, for each cluster $C_{i}$ we are free to choose the starting and ending vertices. We consider the two cases again. In the first case, all required edges in $E^{\prime}$ are only distributed within the clusters. In the second case, there exist  some required edges incident with some different clusters.

For every cluster $C_{i}$, $i\in [k]$, we consider GCRP on the new graph $\overline{G}=\cup \overline{C_{i}}$ defined as before:
$$\overline{C_{i}}=(\overline{V_{i}},\overline{E_{i}})=(\{v|v\in e\in E^{\prime}_{i}\}\cup V^{\prime}_{i}, E^{\prime}_{i}).$$

We first consider the connected components of $\overline{C_{i}}$. In order to obtain the resulted tour, the degree of every vertex of the tour must be even.
Therefore, there also exist some cases that the tour cannot exist, i.e.,  there exists a vertex $v\in\overline{V_{i}}$ with degree $d(v)>2$ (in such a case, at least one of the clusters must be visited more than once).
Henceforth, we will assume that the desired tour exists.

To solve the first case when all required edges in $E^{\prime}$ are only distributed within the clusters, we propose an algorithm which computes two different solutions. Then we select the shorter one of these two tours. To get the first solution, by using Algorithm \ref{alg1} with unspecified ends, we can find paths within each cluster. Then we can view this as a Rural Postman Problem instance. To get the second solution, for each cluster, we select two vertices $s_{i}$ and $t_{i}$ such that $c(s_{i},t_{i})$ is maximized. Let them be the end vertices of each cluster. Then we can apply Algorithm \ref{alg3} to get the second tour. Finally, we select the shorter tour.

The algorithm for the case when the tour exists can be described as follows: 
\begin{algorithm}[H]
	\renewcommand{\algorithmicrequire}{\textbf{Input:}}
    \renewcommand{\algorithmicensure}{\textbf{Output:}}
	\caption{Algorithm of unspecified ending vertices}
	\begin{algorithmic}[1]\label{alg4}
		\REQUIRE
\STATE An edge-weighted graph $G=(V,E,c)$, $V^{\prime}\subseteq V$, $E^{\prime}\subseteq E$.
\STATE A partition of $V$ into clusters $C_{1},\dots,C_{k}$.

        \ENSURE A cluster general routing tour.


\textbf{begin}:\\
\STATE Consider the new graph $\overline{C_{i}}$, for $i\in [k]$.
\STATE Apply Algorithm \ref{alg1} with unspecified end vertices in each cluster $\overline{C_{1}},\dots ,\overline{C_{k}}$. Let path $p_{i}$ be the resulting path on $\overline{C_{i}}$, and denote its end vertices by $a_{i}$ and $b_{i}$. Apply Algorithm Long-Arc 2 and Algorithm Short-Arc 2 to output the shorter solution for RPP with special edges $(a_{i},b_{i})$ and let $T_{1}$ be the tour obtained by replacing special edge $(a_{i},b_{i})$ by path $p_{i}$, for $i\in [k]$.
\STATE In each cluster find vertices $s_{i}$ and $t_{i}$ that maximize $c(s_{i},t_{i})$, for $i\in [k]$.  Apply Algorithm \ref{alg3} with the end vertices $\{s_{i},t_{i}\}$ to output a tour $T_{2}$
\STATE \textbf{Output} the shorter of $T_{1}$ and $T_{2}$.\\
\textbf{end}\\
	\end{algorithmic}
\end{algorithm}


We will analyze the approximation ratio of Algorithm \ref{alg4}. We first introduce some notations. As in the previous section, let $L$ denote the sum of the lengths of the Hamiltonian paths within the clusters in $OPT$, and let $A$ denote the sum of the lengths of the remaining edges of $OPT$. Let $D=\sum^{k}_{i=1}c(s_{i},t_{i})$ denote the sum cost of required edges. The first algorithm works well when $D$ is small, and the second works well when $D$ is large.
\begin{theorem}\label{th7} ($\star$) Let $T_{1}$ be the tour computed in Step 2 of Algorithm \ref{alg4}. Then we have
$$c(T_{1})\leq\frac{3}{2}OPT+\frac{1}{2}L+2D.$$
\end{theorem}

\begin{theorem}\label{th8} ($\star$)
Let $T_{2}$ be the tour computed in Step 3 of Algorithm \ref{alg4}. Then we have
\[
c(T_{2})\leq\frac{3}{2}OPT+3L-2D.
\]
\end{theorem}

Now we can get the following theorem:
\begin{theorem}\label{th14}
Let $T$ be the tour returned by Algorithm \ref{alg4}, then
$$c(T)\leq\frac{13}{4}OPT.$$
\end{theorem}
\begin{proof} Note that $L\leq OPT$. If $2D\leq\frac{5}{4}L$, Theorem \ref{th7} implies that
\[
c(T_{1})\leq\frac{3}{2}OPT+\frac{7}{4}L\leq\frac{13}{4}OPT.
\]
Otherwise, when $2D\geq\frac{5}{4}L$, Theorem \ref{th8} implies that
\[
c(T_{2})\leq\frac{13}{4}OPT.
\]
Since the algorithm chooses the shorter one between the tours $T_{1}$ and $T_{2}$, the proof is completed.\qed
\end{proof}

Next, we will consider Case 2 when there exist required edges between clusters.

We consider the problem in three different cases. In the first case, the number of required edges incident with different clusters is $k$. In the second case, some clusters have two required edges incident with other clusters. In the third case, the number of clusters with required edges incident with other cluster is 0. 

In the first case, we only need to find paths between each specified vertices. This can be seen as an instance of Travelling Salesman Path Problem as we described before. So the performance ratio of this case is 1.5.

In the second case, for the clusters which have two required edges incident with other clusters, we find paths in them and it becomes the third case.

Without loss of generality, we consider the third case: the number of clusters with required edges incident to other cluster is 0. For every cluster $C_{i}$, we denote the specified vertex as $a_{i}$. First, in each cluster, by computing the distance between each component, we select the longest one; that is, we find the vertex $b_{i}$ such that $c(a_{i},b_{i})$ is maximum. This can be done in polynomial time, because the number of vertices in each cluster is no more than $n$. Then we can find the path $p_{i}$ in each cluster $C_{i}$ by Algorithm \ref{alg1}. Since this problem has no direction, we apply Algorithm Long-Arc 2 and Algorithm Short-Arc 2 to output the shorter solution for RPP and find the tour with the edge $(a_{i},b_{i})$. At last, we replace the edge$(a_{i},b_{i})$ by path $p_{i}$. The whole algorithm can be described as follows:
\begin{algorithm}[H]
	\renewcommand{\algorithmicrequire}{\textbf{Input:}}
    \renewcommand{\algorithmicensure}{\textbf{Output:}}
	\caption{Algorithm of existing required edges between clusters}
	\begin{algorithmic}[1]\label{alg5}
		\REQUIRE
\STATE An edge-weighted graph $G=(V,E,c)$, $V^{\prime}\subseteq V$, $E^{\prime}\subseteq E$.
\STATE A partition of $V$ into clusters $C_{1},\dots ,C_{k}$.

        \ENSURE A cluster general routing tour.


\textbf{begin}:\\
\STATE  Let the vertex adjacent to required edges between different cluster $C_{i}$ be $a_{i}$. Find a vertex that maximize $c(a_{i},b_{i})$, for $i=1,\dots,k$.
\STATE For each $\overline{C_{i}}$, compute a path $p_{i}$, a Hamiltonian path with end vertices $a_{i}$ and $b_{i}$, for $i=1,\dots,k$.
\STATE Apply Algorithm Long-Arc 2 and Algorithm Short-Arc 2 to output the shorter solution for RPP with the special edges $\{(a_{i},b_{i})|i=1,\dots,k\}$ to obtain tour $S$, for $i=1,\dots,k$.
\STATE In $T$, replace the special edge $(a_{i},b_{i})$  by the path $p_{i}$, for $i=1,\dots,k$.
\STATE \textbf{return} the resulting tour $T$.\\
\textbf{end}\\
	\end{algorithmic}
\end{algorithm}
\begin{theorem}\label{th13} ($\star$)
 Let $T$ be the tour output by Algorithm \ref{alg5}. Then
$c(T)\leq\frac{9}{4}OPT.$
\end{theorem}
%
%

Algorithm \ref{alg5} is a $2.25$-approximation algorithm for the  general cluster routing problem with unspecified end vertices, in which some required edges may be incident with different clusters. Therefore, the performance ratio of approximation algorithm for the problem with unspecified vertices is $3.25$.
\section{Conclusion}\label{sec:c}
In this paper, we present constant approximation algorithms for two variations of the cluster general routing
problem. However, the two presented algorithms have different approximation ratio, and in our future work we will consider whether we can design approximation algorithms with the same approximation ratio for these two problems.

%
%
%

\newpage
\section*{Appendix}
{\bf Proof of Theorem \ref{th4}}

Let $S$ be the path output by Algorithm \ref{alg1}. Then we have
\begin{align*}
c(S)\le \min\left\{3OPT-c(s,t), \frac{3}{2}OPT+\frac{1}{2}c(s,t)\right\}.
\end{align*}
\begin{proof}
Note that we copy the edges of $MST(G^{*})$ except for those on the $s-t$ path. Therefore $S$ is a connected multigraph whose vertex degrees are all even except for those of $s$ and $t$. In this multigraph we find an Eulerian walk between $s$ and $t$. By the standard Short-cut procedure \cite{Lovas}, we can turn it into a Hamiltonian path between $s$ and $t$ without increasing the weight. Since the edge costs satisfy triangle inequality, we know that the length is at most $2MST(G^{*})-c(s,t)$. For the graph $G^{\prime}$, we can compute the distance between each pair of connected components in polynomial time. So when we retrieve the vertices the connected parts of original graph $G^{\prime}$ from $G^{*}$, we can shortcut the costs of the edges in $E^{\prime}$ without adding cost. Therefore, we have $c(E^{\prime})< OPT$, implying that the length of $S$ is at most $3OPT-c(s,t)$.

Moreover, by adding the edge $s-t$ to $OPT$ to make it a cycle, the length of this cycle is $OPT + c(s,t)$. Decompose the cycle into two matchings. The length of the smaller matching is at most $\frac12(OPT+c(s,t))$. Consider a minimum spanning tree of $G^{*}$. If the vertex is a component of the original graph $G^{\prime}$, then recover the original path. Add a minimum-weight matching of vertices to the minimum spanning tree. Then we obtain an upper bound for any feasible solution. Applying that Short-cut procedure \cite{Lovas} again to get that $c(S)\le MST(G^{\prime})+\frac12(OPT+c(s,t))$ due to the triangle inequality.  So the value of the feasible solution output by Algorithm \ref{alg1} is $\frac{3}{2}OPT+\frac12c(s,t).$
\end{proof}

\noindent{\bf Proof of Theorem \ref{th6}}

 Let $T$ be the tour output by Algorithm \ref{alg3}. Then
$$c(T_{s})\leq2.4\:OPT.$$
\begin{proof} In Step 2, by Algorithm \ref{alg1}, we can generate a path $p_{i}$ in each cluster $\overline{C_{i}}$. In these paths, all required parts in $E^{\prime}$ and $V^{\prime}$ are visited. By Theorem \ref{th4}, the lengths of the two solutions to the travelling salesman path problem with given starting and ending vertices are at most $3L-D$ and $\frac{3}{2}L+D$, respectively. Let $P$ be the solution output by Algorithm \ref{alg1}. Because the minimum of a set of quantities is no more than their convex combination, we have
\begin{align*}
~~~c(P) &\leq \min\left(3L-D,\frac{3}{2}L+\frac{1}{2}D\right)\\
&\leq \frac{3}{5}(3L-D)+\frac{2}{5}\left(\frac{3}{2}L+\frac{1}{2}D\right)\\
&=\frac{12}{5}L-\frac{2}{5}D.
\end{align*}

Since $A$ denotes the length of the other edges of $OPT$ that are not in $L$ and $D$ is the total length of the directed arcs $\overrightarrow{(s_{i},t_{i})},i=1,\dots,k$, it follows that there exists a solution to the Stack Crane Problem of length at most $A+D$. By Theorem \ref{th5}, the lengths of the two solutions to the Stack Crane Problem are at most $\frac{3}{2}A+2D$ and $3A+D$, respectively. Let $S$ be the solution output by the algorithm for the Stack Crane Problem. Because the minimum of a set of quantities is no more than their convex combination, we have
\begin{align*}
~~~c(S) &\leq \min\left(\frac{3}{2}A+2D,3A+D\right)\\
&\leq \frac{2}{5}\left(\frac{3}{2}A+2D\right)+\frac{3}{5}(3A+D)\\
&=\frac{12}{5}A+\frac{7}{5}D.
\end{align*}
Finally, we combine the two solutions by replacing arcs of length $D$ in the travelling salesman path problem solution by the shorter solution output by Algorithm Long-Arc 1 and Algorithm Short-Arc 1 for the Stack Crane Problem. We obtain an upper bound on the length of the solution $T_{s}$ by combining the above inequalties.
\begin{align*}
~~~c(T_{s}) &= c(P)-D+c(S)\\
&\leq \frac{12}{5}L-\frac{2}{5}D-D+\frac{12}{5}A+\frac{7}{5}D\\
&=\frac{12}{5}(L+A)=\frac{12}{5}OPT.
\end{align*}
The proof is completed.
\end{proof}

\noindent{\bf Proof of Theorem \ref{th7}}

Let $T_{1}$ be the tour computed in Step 2 of Algorithm \ref{alg4}. Then we have
$$c(T_{1})\leq\frac{3}{2}OPT+\frac{1}{2}L+2D.$$

\begin{proof} We first consider an optimal solution $OPT$. Orientate its edges arbitrarily. In $OPT$, let $u_{i}$ and $v_{i}$ be the first and last vertices of cluster $\overline{C_{i}}$. Without loss of generality, we assume that $a_{i}$ and $b_{i}$ are two vertices in $OPT$ and $u_{i}$, $a_{i}$, $b_{i}$, $v_{i}$ are in this order. Then for all $k$ clusters, by summing them up, we can get
$$L\geq\sum^{k}_{i=1}\Big{(}c(u_{i},a_{i})+c(a_{i},b_{i})+c(b_{i},v_{i})\Big{)}.$$

Since $A$ is the sum of the lengths of the remaining edges of $OPT$, i.e., $A$ denotes the edge links between clusters, there exists a rural postman tour with special edges $(a_{i},b_{i})$ of length at most
$$A+\sum^{k}\limits_{i=1}\Big{(}c(u_{i},a_{i})+c(a_{i},b_{i})+c(b_{i},v_{i})\Big{)}.$$

The approximation ratio of the algorithm for RPP is 1.5 from Remark \ref{re1} earlier. Then after applying the algorithm, the length of the tour is at most
$$1.5\bigg{(}A+\sum^{k}_{i=1}\Big{(}c(u_{i},a_{i})+c(a_{i},b_{i})+c(b_{i},v_{i})\Big{)}\bigg{)}.$$

Since $\sum^{k}\limits_{i=1}c(p_{i})\leq\frac{3}{2}L$, we replace the special edge $(a_{i},b_{i})$ by the path $p_{i}$ to obtain
\begin{align*}
~~~c(T_{1}) &\leq\frac{3}{2}\bigg{(}A+\sum^{k}_{i=1}\Big{(}c(u_{i},a_{i})+c(a_{i},b_{i})+c(b_{i},v_{i})\Big{)}\bigg{)}-\sum^{k}_{i=1}c(a_{i},b_{i})+\frac{3}{2}L\\
&\leq\frac{3}{2}OPT+\frac{1}{2}\sum^{k}_{i=1}\Big{(}c(u_{i},a_{i})+c(a_{i},b_{i})+c(b_{i},v_{i})\Big{)}+\sum^{k}_{i=1}\Big{(}c(u_{i},a_{i})+c(b_{i},v_{i})\Big{)}\\
&\leq\frac{3}{2}OPT+\frac{1}{2}\sum^{k}_{i=1}\Big{(}c(u_{i},a_{i})+c(a_{i},b_{i})+c(b_{i},v_{i})\Big{)}+2D\\
&\leq\frac{3}{2}OPT+\frac{1}{2}L+2D.
\end{align*}
\end{proof}

\noindent{\bf Proof of Theorem \ref{th8}}

Let $T_{2}$ be the tour computed in Step 3 of Algorithm \ref{alg4}. Then we have
\[
c(T_{2})\leq\frac{3}{2}OPT+3L-2D.
\]

\begin{proof} For $T_{2}$, we first find two vertices $s_{i}$ and $t_{i}$ that maximize $c(s_{i},t_{i})$ and this can be done in polynomial time. This process can be viewed as a Rural Postman Problem instance. We can find a solution for RPP which is at most $\frac{3}{2}OPT$. Then, we replace each special edge $(s_{i},t_{i})$ by a path connecting $s_{i}$ and $t_{i}$, that includes all required vertices in $\overline{C_{i}}$. The length of these paths is bounded as follows
\[
3\sum^{k}\limits_{i=1}MST(\overline{C_{i}})-D\leq 3L-D.
\]
 Hence the length of the tour is at most
 \[\frac{3}{2}OPT-D+(3L-D)=\frac{3}{2}OPT+3L-2D.
 \]
\end{proof}

\noindent{\bf Proof of Theorem \ref{th13}}

 Let $T$ be the tour output by Algorithm \ref{alg5}. Then
$$c(T)\leq\frac{9}{4}OPT.$$

\begin{proof}  The length of the paths computed in Step 1 is at most:
\begin{align*}
~~~c(P) &\leq \min\left\{3L-D,\frac{3}{2}L+\frac{1}{2}D\right\}\\
&\leq \frac{1}{2}(3L-D)+\frac{1}{2}\left(\frac{3}{2}L+\frac{1}{2}D\right)\\
&=\frac{9}{4}L-\frac{1}{4}D.
\end{align*}

Since $A+D=OPT$, there is a solution output by Algorithm Long-Arc 2 and Algorithm Short-Arc 2 for RPP of length at most $A+D$. The performance ratio is 1.5 from Remark~\ref{re1}. So the lengths of the two solutions we can find for RPP  are $\frac{3}{2}(A+D)$ and $3A+2D$, respectively. Therefore in Step 2, the solution returned by RPP algorithm is bounded as follows:
\begin{align*}
~~~c(S) &\leq \min\left\{\frac{3}{2}(A+D),3A+D\right\}\\
&\leq \frac{1}{2}\frac{3}{2}(A+D)+\frac{1}{2}(3A+D)\\
&=\frac{9}{4}A+\frac{5}{4}D.
\end{align*}

In Step 3 of Algorithm \ref{alg5}, the two solutions are merged by replacing edges of length $D$ in the solution of TGPP with that of RPP. We thus obtain an upper bound on the length of the solution $T$ by the above inequalities:
\begin{align*}
~~~c(T) &= c(P)-D+c(S)\\
&\leq \frac{9}{4}L-\frac{1}{4}D-D+\frac{9}{4}A+\frac{5}{4}D\\
&=\frac{9}{4}(L+A)=\frac{9}{4}OPT.
\end{align*}
\end{proof}
\end{document}